\documentclass[letterpaper]{article}
\usepackage{arxiv}
\usepackage{times}
\usepackage{helvet}
\usepackage{courier}
\def\year{2021}\relax

\usepackage{booktabs} 
\usepackage[ruled]{algorithm2e} 
\usepackage{xcolor}
\usepackage{natbib}
\usepackage{amsmath,amssymb,amsfonts}
\usepackage{amsthm}
\usepackage{booktabs}
\usepackage{graphicx}
\usepackage{textcomp}
\usepackage[title]{appendix}

\ifodd 1

\newcommand{\com}[1]{{\color{red}\textbf{Comment}:#1}}
\else

\newcommand{\com}[1]{}
\fi

\newtheorem{definition}{Definition}
\newtheorem{theorem}{Theorem}

\newtheorem{assumption}{Assumption}

\SetAlFnt{\small}
\SetAlCapFnt{\small}
\SetAlCapNameFnt{\small}
\SetAlCapHSkip{0pt}
\IncMargin{-\parindent}

\begin{document}
%

\title{Multi-Scale Games: Representing and Solving Games on Networks \\
with Group Structure}
\author{Kun Jin\\
{University of Michigan, Ann Arbor}\\
kunj@umich.edu
\And
Yevgeniy Vorobeychik\\
{Washington University in St. Louis}\\
yvorobeychik@wustl.edu
\And
Mingyan Liu\\
{University of Michigan, Ann Arbor}\\
mingyan@umich.edu}
\maketitle
\begin{abstract}
Network games provide a natural machinery to compactly represent strategic interactions among agents whose payoffs exhibit sparsity in their dependence on the actions of others.
Besides encoding interaction sparsity, however, real networks often exhibit a multi-scale structure, in which agents can be grouped into communities, those communities further grouped, and so on, and where interactions among such groups may also exhibit sparsity.
We present a general model of multi-scale network games that encodes such multi-level structure.
We then develop several algorithmic approaches that leverage this multi-scale structure, and derive sufficient conditions for convergence of these to a Nash equilibrium.
Our numerical experiments demonstrate that the proposed approaches enable orders of magnitude improvements in scalability when computing Nash equilibria in such games. For example, we can solve previously intractable instances involving up to 1 million agents in under 15 minutes.
\end{abstract}

\section{Introduction}

Strategic interactions among interconnected agents are commonly modeled using the network, or graphical, game formalism~\citep{Kearns01,jackson2015games}.
In such games, the utility of an agent depends on his own actions as well as those by its network neighbors.
Many variations of games on networks have been considered, with applications including the provision of public goods~\citep{allouch2015private,buckley2006income,khalili2019public,Yu20},  security~\citep{hota2018interdependent,la2016interdependent,Vorobeychik15}, and financial markets~\citep{acemoglu2012network}.

\begin{figure}[htb] 
\centering 
\includegraphics[width=0.45\textwidth]{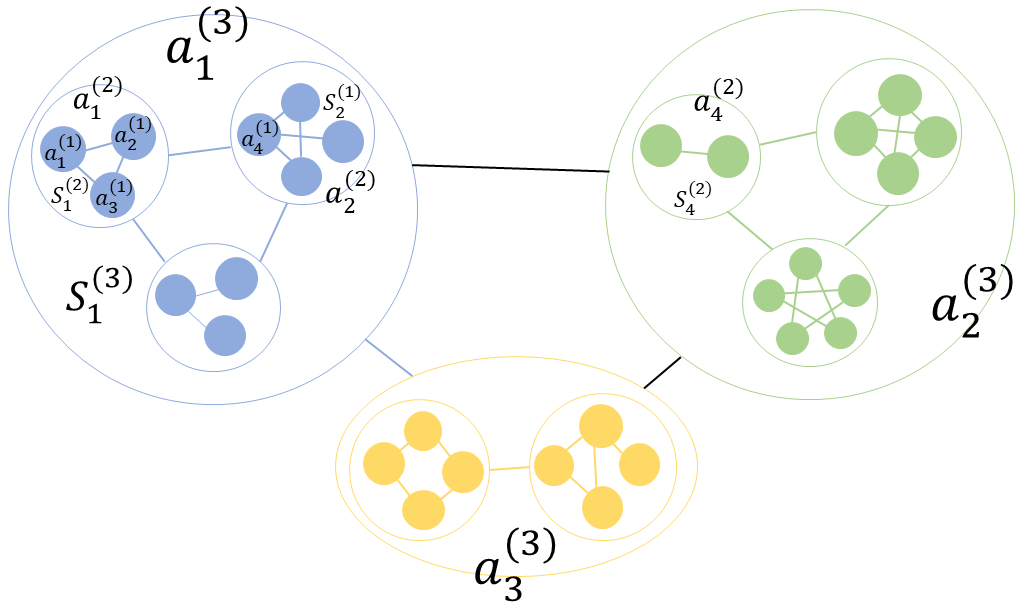}
\caption{An illustration of a multi-scale (3-level) network.} 
\label{fig:multi_scale}
\end{figure}
While network games are a powerful modeling framework, they fail to capture a common feature of human organization: groups and communities.
Indeed, investigation of communities, or close-knit groups, in social networks is a major research thread in network science.
Moreover, such groups often have a hierarchical structure~\citep{Clauset08,Girvan02}.
For example, strategic interactions among organizations in a marketplace often boil down to interactions among their constituent business units, which are, in turn, comprised of individual decision makers.
In the end, it is those lowest-level agents who ultimately accrue the consequences of these interactions (for example, corporate profits would ultimately benefit individual shareholders).
Moreover, while there are clear interdependencies among organizations, individual utilities are determined by a combination of individual actions of some agents, together with \emph{aggregate} decisions by the groups (e.g., business units, organizations).
For example, an employee's bonus is determined in part by their performance in relation to their co-workers, and in part by how well their employer (organization) performs against its competitors in the marketplace.

We propose a novel \emph{multi-scale game model} that generalizes network games to capture such hierarchical organization of individuals into groups.
Figure~\ref{fig:multi_scale} offers a stylized example in which three groups (e.g., organizations) are comprised of 2-3 subgroups each (e.g., business units), which are in turn comprised of 2-5 individual agents.
Specifically, our model includes an explicit hierarchical network structure that organizes agents into groups across a series of levels.
Further, each group is associated with an action which deterministically aggregates the decisions by its constituent agents.
The game is grounded at the lowest level, where the agents are associated with scalar actions and utility functions that have modular structure in the strategies taken at each level of the game.
For example, in Figure~\ref{fig:multi_scale}, the utility function of an individual member $a_j$ of level-3 group $a_3^{(3)}$ is a function of the strategies of (i) $a_j$'s immediate neighbors (represented by links between pairs of filled-in circles), (ii) $a_j$'s level-2 group and its network neighbor (the small hollow circles), and (iii) $a_j$'s level-3 group, $a_3^{(3)}$ (large hollow circle) and its network neighbors, $a_1^{(3)}$ and $a_2^{(3)}$.

%
Our next contribution is a series of iterative algorithms for computing pure strategy Nash equilibria that explicitly leverage the proposed multi-scale game representation.
The first of these simply takes advantage of the compact game representation in computing equilibria.
The second algorithm we propose offers a further innovation through an iterative procedure that alternates between game levels, treating groups themselves as pseudo-agents in the process.
We present sufficient conditions for the convergence of this algorithm to a pure strategy Nash equilibrium through a connection to Structured Variational Inequalities  \citep{He2000ADMAdap}, although the result is limited to games with two levels.
To address the latter limitation, we design a third iterative algorithm that now converges even in games with an arbitrary number of levels.

Our final contribution is an experimental evaluation of the proposed algorithms compared to best response dynamics.
In particular, we demonstrate orders of magnitude improvements in scalability, enabling us to solve games that cannot be solved using a conventional network game representation.



\smallskip
\noindent{\bf Related Work:} 
Network games have been an active area of research; see e.g., surveys by~\citet{jackson2015games} and \citet{GPoN2016}. We now review the most relevant papers.
%
Conditions for the existence, uniqueness and stability of Nash equilibria in network games under general best responses are studied in \citep{parise2019variational, naghizadeh2017uniqueness, scutari2014real, bramoulle2014strategic}. Variational inequalities (VI) are used in these works to analyze the fixed point and contraction properties of the best response mappings. It is identified in \cite{parise2019variational, naghizadeh2017uniqueness, scutari2014real} that when the Jacobian matrix of the best response mapping is a P-matrix or is positive definite, a feasible unique Nash equilibrium exists and can be obtained by best-response dynamics \citep{scutari2014real, parise2019variational}. 
In this paper, we extended the analysis of equilibrium and best responses for a conventional network game to that in a multi-scale network game, where the utility functions are decomposed into separable utility components to which  best responses are applied separately. This is similar to the generalization from a conventional VI problem to an SVI problem \citep{He2000ADMAdap, He2009ParaALM, He2012ADMConv, Bnouhachem2013InexADM} problem.  


Previous works on network games that involve group or community structure focus on finding such structures; e.g., community detection in networks using game theoretic methods have been studied in \citep{GTFCD2017, DetecStru2004, CDDSN2014}. 
By contrast, our work focuses on analyzing a network game with a given  group/community structure, and using the structure as an analytical tool for the analysis of equilibrium and best responses. 
\section{Preliminaries} \label{sec:preliminaries}


A general \emph{normal-form game} is defined by a set of agents (players) $I = \{1,\ldots,N\}$, with each agent $a_i$ having an action/strategy space $K_i$ and a utility function $u_i(x_i,\pmb{x}_{-i})$ that $i$ aims to maximize; $x_i \in K_i$ and $x_{-i}$ denotes the actions by all agents other than $i$.
We term the collection of strategies of all agents $\pmb{x}$ a \emph{strategy profile}.
We assume $K_i \subset \mathbb{R}$ is a compact  set. 

We focus on computing a \emph{Nash equilibrium (NE)} of a normal-form game, which is a strategy profile with each agent maximizing their utility given the strategies of others.
Formally, $\pmb{x}^*$ is a \emph{Nash equilibrium} if for each agent $i$,
\begin{equation} \label{eqn:NE}
    x_i^* \in \underset{x_i \in K_i}{\text{argmax}}~~ u_i(x_i, \pmb{x}_{-i}^*). 
\end{equation}

A \emph{network game} encodes structure in the utility functions such that they only depend on the actions by network neighbors.
Formally, a network game is defined over a weighted graph $(I,E)$, with each node an agent and $E$ is the set of edges; the agent's utility $u_i(x_i, \pmb{x}_{-i})$ reduces to $u_i(x_i,\pmb{x}_{I_i})$, where $I_i$ is the set of network neighbors of $i$, although we will frequently use the former for simplicity.

An agent's best response is its best strategy 
given the actions taken by all the other agents.  
Formally, the best response is a set defined by
\begin{equation} \label{eqn:BR}
BR_i(\pmb{x}_{-i}, u_i) = \underset{x_i}{\text{argmax}} ~u_i(x_i, \pmb{x}_{-i}).
\end{equation}
Whenever we deal with games that have a unique best response, we will use the singleton best response set to also refer to the player's best response strategy (the unique member of this set).

Clearly, a NE of a game is a fixed point of this best response correspondence.
Consequently, one way to compute a NE of a game is through \emph{best response dynamics (BRD)}, which
is a process whereby agents iteratively and asynchronously (that is, one agent at a time) take the others' actions as fixed values and play a best response to them.
%

We are going to use this BRD algorithm as a major building block below.
One important tool that is useful for analyzing BRD convergence is \emph{Variational Inequalities (VI)}.
To establish the connection between NE and VI we assume the utility functions $u_i, \forall i = 1,\dots,N$, are continuously twice differentiable. 
Let $K = \prod_{i=1}^N K_i$ and define $F: \mathbb{R}^N \rightarrow \mathbb{R}^N$ as follows:
\begin{equation} \label{eqn:F(x)}
	F(\pmb{x}) := \bigg( -\triangledown_{x_i} u_i(\pmb{x}) \bigg)_{i=1}^N ~.  
\end{equation}
Then $\pmb{x}^*$ is said to be a solution to VI$(K, F)$ if and only if
\begin{equation} \label{eqn:VI_flat}
	(\pmb{x} - \pmb{x}^*)^T F(\pmb{x}^*) \geq 0, ~  \forall \pmb{x} \in K ~.
\end{equation}
In other words, the solution set to VI$(K,F)$ is equivalent to the set of NE of the game.
Now, we can define the condition that will guarantee the convergence of BRD.
\begin{definition} \label{thm:P_upsilon}
	\textbf{The $P_{\Upsilon}$ condition}: The $\Upsilon$ matrix generated from $F: \mathbb{R}^N \rightarrow \mathbb{R}^N$ is given as follows
	\begin{equation*} 
	\Upsilon(F) = \begin{bmatrix}
	\alpha_1(F) & -\beta_{1,2}(F) & \cdots & -\beta_{1,N}(F) \\
	-\beta_{2,1}(F) & \alpha_2(F) & \cdots & -\beta_{2,N}(F) \\
	\vdots & \vdots & \ddots & \vdots \\
	-\beta_{N,1}(F) & -\beta_{N,2}(F) & \cdots & \alpha_N(F)
	\end{bmatrix},
	\end{equation*}
$	    \alpha_i(F) = \inf_{\pmb{x} \in K} ||\triangledown_i F_i||_2$, $\beta_{i,j}(F) = \sup_{\pmb{x} \in K} ||\triangledown_j F_i||_2$, $i \neq j$. 
	If $\Upsilon(F)$ is a P-matrix, that is, if all of its principal components have a positive determinant, then we say $F$ satisfies the $P_{\Upsilon}$ condition. 
\end{definition}

\begin{theorem}\citep{scutari2014real}
\label{T:BRD}
    If $F$ satisfies the $P_{\Upsilon}$ condition, then $F$ is strongly monotone on $K$, and VI$(K,F)$ has a unique solution.  Moreover, BRD converges to the unique NE from an arbitrary initial state.
\end{theorem}


\section{A Multi-Scale Game Model} \label{sec:model}

Consider a conventional network (graphical) game with the set $I$ of $N$ agents situated on a network $G=(I,E)$, each with a utility function $u_i(x_i,\pmb{x}_{I_i})$, with $I_i$ the set of $i$'s neighbors, $I$ the full set of agents/nodes and $E$ the set of edges connecting them.\footnote{The edges are generally weighted, resulting in a weighted adjacency matrix on which the utility depends.} 
Suppose that this network $G$ exhibits the following structure and feature of the strategic dependence among agents: agents can be partitioned into a collection of groups $\{S_k\}$, where $k$ is a group index, and an agent $a_i$ in the $k$th group (i.e., $a_i \in S_k$) has a utility function that depends (i) on the strategies of its network neighbors in $S_k$, and (ii) \emph{only on the aggregate strategies} of groups other than $k$ (see, e.g., Fig.~\ref{fig:multi_scale}). 
Further, these groups may go on to form larger groups, whose aggregate strategies impact each other's agents, giving rise to a {\em multi-scale} structure of the network. 
This kind of structure is very natural in a myriad of situations.
For example, members of criminal organizations take stock of individual behavior by members of their own organization, but their interactions with other organizations (criminal or otherwise) are perceived in group terms (e.g., how much another group has harmed theirs).
A similar multi-level interaction structure exists in national or ethnic conflicts, organizational competition in a market place, and politics. 
Indeed, a persistent finding in network science is that networks exhibit a multi-scale interaction structure (i.e., communities, and hierarchies of communities)~\citep{Girvan02,Clauset08}.

We present a general model to capture such multi-scale structure.  
Formally, an $L$-level structure is given by a hierarchical graph structure $\{G^{(l)}\}$ for each level $l$, $1\leq l \leq L$, where $G^{(l)}= (\{S_k^{(l)}\}_{k}, E^{(l)})$ represents the level-$l$ structure. The first component,   $\{S_k^{(l)}\}_k$ prescribes a partition, where agents in level $l-1$ form disjoint groups given by this partition; each group is viewed as an agent in level $l$, denoted as $a^{(l)}_k$.  Notationally, while both $a^{(l)}_k$ and $S_k^{(l)}$ bear the superscript $(l)$, the former refers to a level-$l$ agent, while the latter is the group (of level-$(l-1)$ agents) that the former represents. The set of level-$l$ agents is denoted by $I^{(l)}$ and their total number $N^{(l)}$. 
The second component, $E^{(l)}$, is a set of edges that connect level-$l$ agents, encoding the dependence relationship among the groups they represent. 
This structure is anchored in level 1 (the lowest level), where sets $S_k^{(1)}$ are singletons, corresponding to agents $a_k$ in the game, who constitute the set $I$.

To illustrate, the multi-scale structure shown in Fig.~\ref{fig:multi_scale} is given by $G^{(1)} = G =(\{S_k^{(1)}\}_k=I, E^{(1)}=E)$, as well as how level-1 agents are grouped into level-2 agents, how level-2 agents are further grouped into level-3 agents, and the edges connecting these groups at each level. 

It should be obvious that the above multi-scale representation of a graphical game is a generalization of a conventional graphical game, as any such game essentially corresponds to a $L=1$ multi-scale representation.  On the other hand, not all conventional graphical games have a meaningful $L>1$ multi-scale representation (with non-singleton groups of level-1 agents); this is because our assumption that an agent's utility only depends on the \emph{aggregate} decisions by groups other than the one they belong to implies certain properties of the dependence structure.  For the remainder of this paper we will proceed with a given multi-scale structure defined above, while in Appendix \ref{appendix:reverse_engineer} we outline a set of conditions on a graphical game $G$ that allows us to represent it in a (non-trivial) multi-scale fashion.

Since the resulting multi-scale network is strictly hierarchical, we can define a \emph{direct supervisor} of agent $a_i^{(l)}$ in level-$l$ to be the agent $a_k^{(l+1)}$ corresponding to the level-($l+1$) group $k$ that the former belongs to.
Similarly, two agents who belong in the same level-$l$ group $k$ are (level-$l$) \emph{group mates}.
Finally, note that any level-1 agent $a_i$ belongs to exactly one group in each level $l$.
We index a level-$l$ group to which $a_i$ belongs by $k_{il}$.

In order to capture the agent dependence on aggregate actions, we define an \emph{aggregation function} $\sigma_k^{(l)}$ for each level-$l$ group $k$ that maps individual actions of group members to $\mathbb{R}$ (a \emph{group strategy}).
Specifically, consider a level-$l$ group $S_k^{(l)}$ with level-($l-1$) agents in this group playing a strategy profile $\pmb{x}_{S_k^{(l)}}$.
The (scalar) group strategy, which is also the strategy for the corresponding level-($l+1$) agent, is determined by the aggregation function, 
\begin{align}
    \label{eqn:agg_action}
    x_k^{(l)} = \sigma_k^{(l)}(\pmb{x}_{S_k^{(l)}}).
\end{align}
A natural example of this is linear (e.g., agents respond to total levels of violence by other criminal organizations): $\sigma_k^{(l)}(\pmb{x}_{S_k^{(l)}}) = \sum_{i \in S_k^{(l)}} x_i^{(l)}$.

The $L$-level structure above is captured strategically by introducing structure into the utility functions of agents.
Let $I_{k_{il}}$ denote the set of neighbors of level-$l$ group $k$ to which level-1 agent $a_i$ belongs; i.e., this is the set of level-$l$ groups that interact with agent $a_i$'s group.
This level-1 agent's utility function can be decomposed as follows:
\begin{equation} \label{eqn:utility}
    u_i(x_i, \pmb{x}_{-i}) = \sum_{l=1}^{L} u^{(l)}_{k_{il}} \bigg( x^{(l)}_{k_{il}}, \pmb{x}^{(l)}_{I_{k_{il}}} \bigg).
\end{equation}
In this definition, the level-$l$ strategies $x^{(l)}_{k}$ are implicitly functions of the level-1 strategies of agents that comprise the group, per a recursive application of Eqn.~\eqref{eqn:agg_action}.
Consequently, the utility is an additive function of the hierarchy of group-level components for increasingly (with $l$) abstract group of agents.
Note that conventional network games are a special case with only a single level ($L=1$).

To illustrate, if we consider just two levels (a collection of individuals and groups to which they directly belong), the utility function of each agent $a_i$ is a sum of two components:
\[
    u_i(x_i, \pmb{x}_{-i}) =  u^{(1)}_{k_{i1}} \bigg( x^{(1)}_{k_{i1}}, \pmb{x}^{(1)}_{I_{k_{i1}}} \bigg) + u^{(2)}_{k_{i2}} \bigg( x^{(2)}_{k_{i2}}, \pmb{x}^{(2)}_{I_{k_{i2}}} \bigg).
\]
In the first component, $x^{(1)}_{k_{i1}} = x_i$, since level-1 groups correspond to individual agents, whereas $\pmb{x}^{(1)}_{I_{k_{i1}}}$ is the strategy profile of $i$'s  neighbors \emph{belonging to the same group as $i$}, given by  $E^{(1)}$.
The second utility component now depends only on the aggregate strategy $x^{(2)}_{k_{i2}}$ of the group to which $i$ belongs, as well as the aggregate strategies of the groups with which $i$'s group interacts, given by $E^{(2)}$.

\section{Algorithms and Analysis}\label{sec:algorithms} 

Consider the BRD algorithm (formalized in Algorithm~\ref{alg:flat_BRD}) in which we iteratively select an agent who plays a best response to the strategy of the rest from the previous iteration.
\begin{algorithm}[htbp]\label{alg:flat_BRD}
	Initialize the game, $t=0, x_i(0) = (\pmb{x}_0)_i, i=1,\cdots,N$\;
	\While{not converged}{
		\For{i = 1:$N$} 
		{
			$x_i(t+1) = BR_i (\pmb{x}_{-i}(t), u_i)$\\
		}
		$t \leftarrow t+1$
	}
	\caption{BRD Algorithm}
\end{algorithm}

The conventional BRD algorithm operates on the ``flattened'' utility function which evaluates utilities explicitly as functions of the strategies played by all agents $a_i \in I$.
Our goal henceforth is to develop algorithms that take advantage of the special multi-scale structure and enable significantly better scalability than standard BRD, while preserving the convergence properties of BRD.





\subsection{Taking Advantage of Multi-Scale Utility Representation}

The simplest way to take advantage of the multi-scale representation is to directly leverage the structure of the utility function in computing best responses.
Specifically, the multi-scale utility function is more compact than one that explicitly accounts for the strategies of all  neighbors of $i$ (which includes \emph{all} of the players in groups other than the one $i$ belongs to).
This typically results in a direct computational benefit to computing a best response.
For example, in a game with a linear best response, this can result in an exponential reduction in the number of linear operations.

\begin{algorithm}[htbp]\label{alg:hierarchical}
	\SetAlgoLined
	
	Initialize the game, $t=0, x^{(1)}_i(0)=(\pmb{x}_0)_i, i = 1,\dots,N$
	
	\For{l = 2:L}{
	    \For{k = 1:$N^{(l)}$}{
		    $\pmb{x}^{(l)}_k(0) = \sigma^{(l)}_k ( \pmb{x}_{S^{(l)}_k}(0) )$;\\
		}
	}
	\While{not converged}{
	\For{i = 1:$N$ (Level-1)}{
		$x^{(1)}_i(t + 1) = BR_i (\pmb{x}^{(1)}_{-i}(t), u_i)$\\
	}
	\For{l = 2:L}{	
		\For{k = 1:$N^{(l)}$}{
		    $\pmb{x}^{(l)}_k(t+1) = \sigma^{(l)}_k ( \pmb{x}_{S^{(l)}_k}(t+1) )$;\\
		}
	}
		$t \leftarrow t+1$;
	}
	\caption{Multi-Scale BRD (MS-BRD)} 
\end{algorithm}
The resulting algorithm, \emph{Multi-Scale Best-Response Dynamics (MS-BRD)}, which takes advantage of our utility representation is formalized as Algorithm~\ref{alg:hierarchical}.
The main difference from BRD is that it explicitly uses the multi-scale utility representation:  in each iteration, it updates the aggregated strategies at all levels for the groups to which the most recent best-responding agent belongs.
Since \textsc{MS-BRD} simply performs operations identical to BRD but efficiently, its convergence is guaranteed under the same conditions (see Theorem~\ref{T:BRD}).
Next, we present iterative algorithms for computing NE that take further advantage of the multi-scale structure, and study their convergence.

\subsection{Taking Advantage of Multi-Scale Strategic Dependence Structure}

In order to take full advantage of the multi-scale game structure, we now aim to develop algorithms that treat groups explicitly as agents, with the idea that iterative interactions among these can significantly speed up convergence.
Of course, in our model groups are not actual agents in the game: utility functions are only defined for agents in level 1.
However, note that we already have well-defined group strategies -- these are just the aggregations of agent strategies at the level immediately below, per the aggregation function~\eqref{eqn:agg_action}.
Moreover, we have natural utilities for groups as well: we can use the corresponding group-level component of the utility of any agent in the group (note that these are identical for all group members in Eqn.~\eqref{eqn:utility}).
However, using these as group utilities will in fact not work: since ultimately the game is only among the agents in level 1, equilibria of all of the games at more abstract levels \emph{must be consistent with equilibrium strategies in level 1}.
On the other hand, we need to enforce consistency only between neighboring levels, since that fully captures the across-level interdependence induced by the aggregation function.
Therefore, we define the following \emph{pseudo-utility functions} for agents at levels other than 1, with agent $k$ in level $l$ corresponding to a subset of agents from level $l-1$:
\begin{align} \label{eqn:level_l}
    \hat{u}^{(l)}_k = &~ u^{(l)}_k \bigg( x^{(l)}_k, \pmb{x}^{(l)}_{I_k} \bigg) \nonumber 
    - L_k^{(l,l-1)} \bigg( x^{(l)}_k, \sigma^{(l)}_k ( \pmb{x}_{S^{(l)}_k} ) \bigg) \nonumber \\
    &~ - L_k^{(l,l+1)} \bigg( \sigma^{(l+1)}_k ( \pmb{x}_{S^{(l+1)}_k} ), x^{(l+1)}_k \bigg).
\end{align}
The first term is the level-$l$ component of the utility of any level-1 agent in group $k$.
The second and third terms model the inter-level inconsistency loss that penalizes a level-$l$ agent $a^{(l)}_k$, where $L_k^{(l,l+1)}$ and $L_i^{(l,l-1)}$ penalize its inconsistency with the level-$(l+1)$ and level-$(l-1)$ entities respectively. In general, $L_k^{(l,l+1)}$ is a different function from $L_k^{(l+1,l)}$; we elaborate on this further below.

The central idea behind the second algorithm we propose is simple: in addition to iterating best response steps at level 1, we now interleave them with best response steps taken by agents at higher levels, which we can since strategies and utilities of these pseudo-agents are well defined.
This algorithm is similar to the augmented Lagrangian method in optimization theory, where penalty terms are added to relax an equality constraint and turn the problem into one with separable operators. We can decompose this type of problem into smaller subproblems and solve the subproblems sequentially using the alternating direction method (ADM) \citep{Yuan2011LQPSVI,Bnouhachem2013InexADM}. 
The games at adjacent levels are coupled through the equality constraints on their action profiles given by Eqn (\ref{eqn:agg_action}), and the penalty functions are updated before starting a new iteration. 
The full algorithm, which we call \emph{Separated Hierarchical BRD (SH-BRD)}, is provided in Algorithm~\eqref{alg:separated}.





The penalty updating rule in iteration $t$ of Algorithm~\eqref{alg:separated} is:
\begin{enumerate}
    \item For $l = 2,\dots, L, i=1,\dots,N^{(l)}$
	\begin{align} \label{eqn:loss2_update}
	    &~ L_i^{(l,l-1)} \bigg( x^{(l)}_i, \sigma^{(l)}_i ( \pmb{x}_{S^{(l)}_i}(t+1) ) \bigg) \nonumber \\
	    = &~ h^{(l)}_i \bigg[ x^{(l)}_i - \sigma^{(l)}_i ( \pmb{x}_{S^{(l)}_i}(t+1) ) + \lambda^{(l)}_i(t) \bigg]^2.
	\end{align}

	\item For $l = 1,\dots, L-1; i=1,\dots,N^{(l)}$, where $a_i^{(l)} \in S_k^{(l+1)}$ 
	\begin{align} \label{eqn:loss1_update}
	    &~ L_k^{(l,l+1)} \bigg( \sigma^{(l+1)}_k ( \pmb{x}_{S^{(l+1)}_k}) , x^{(l+1)}_k(t) \bigg) \nonumber\\
	    = &~ h^{(l+1)}_k \bigg[ \sigma^{(l+1)}_k ( \pmb{x}_{S^{(l+1)}_k} ) - x^{(l+1)}_k(t) - \lambda^{(l+1)}_k(t) \bigg]^2.
	\end{align}
	
	\item For $l = 2,\dots, L, i=1,\dots,N^{(l)}$ 
	\begin{align} \label{eqn:lambda_update}
	    &~ \lambda^{(l)}_i(t+1) \nonumber \\
	    = &~ \lambda^{(l)}_i(t) - h^{(l)}_i \bigg[ \sigma^{(l)}_i ( \pmb{x}_{S^{(l)}_i}(t+1) ) - x^{(l)}_i(t+1) \bigg].
	\end{align}
\end{enumerate}
When updating, all other variables are treated as fixed, and $\pmb{\lambda}^{(l)}(0)$, $h^{(l)}_i > 0$ are chosen arbitrarily.

\begin{algorithm}[htbp]\label{alg:separated}
	\SetAlgoLined
	
	Initialize the game, $t=0, x^{(1)}_i(0)=(\pmb{x}_0)_i, i = 1,\dots,N^{(0)}$
	
	\For{l = 2:L}{
	    \For{k = 1:$N^{(l)}$}{
		    $\pmb{x}^{(l)}_k(0) = \sigma^{(l)}_k ( \pmb{x}_{S^{(l)}_k}(0) )$;\\
		}
	}
	\While{not converged}{
	    \For{l = 1:L}{
    		\For{i = 1:$N^{(l)}$ ($l$ to $l-1$ Penalty Update, if $l > 1$)}{
    			Update $L_i^{(l, l-1)}$
    		}
    		\For{i = 1:$N^{(l)}$ ($l$ to $l+1$ Penalty Update, if $l < L$)}{
    			Update $L_k^{(l, l+1)}$, where $a_i^{(l)} \in S^{(l+1)}_k$
    		}
    		\For{i = 1:$N^{(l)}$ (Best Response)}{
    			\begin{align*}
    			    x^{(l)}_i(t + 1) = BR_i \bigg(& \sigma^{(l)}_i ( \pmb{x}_{S^{(l)}_i}(t+1) ), \\
    			    &\pmb{x}^{(l)}_{I_i}(t), x^{(l+1)}_k(t), \hat{u}^{(l)}_i \bigg) 
    			\end{align*}\\
    		}
		}
		$t \leftarrow t+1$;
	}
	\caption{Separated Hierarchical BRD (SH-BRD)} 
\end{algorithm}

Unlike \textsc{MS-BRD}, the convergence of the \textsc{SH-BRD} algorithm is non-trivial.
To prove it, we exploit a connection between this algorithm and Structured Variational Inequalities (SVI) with separable operators~\citep{He2009ParaALM, He2012ADMConv, Bnouhachem2013InexADM}.
To formally state the convergence result, we need to make several explicit assumptions.
\begin{assumption} \label{assumption:cont_diff}
	The functions $u^{(l)}_i, \forall l = 1,\dots,L, \forall i = 1,\dots,N^{(l-1)}$ are twice continuously differentiable.
\end{assumption}

\begin{assumption}\label{assumption:monotone}
	$-\triangledown_{x^{(l)}_i} u^{(l)}_i$ are monotone $\forall l = 1,\dots,L, \forall i = 1,\dots,N^{(l-1)}$. The solution set of $\triangledown_{x^{(l)}_i} u^{(l)}_i = 0, \forall l = 1,\dots,L, \forall i = 1,\dots,N^{(l-1)}$ is nonempty, 
	with solutions in the interior of the action spaces.
\end{assumption}



Let $F^{(l)}$ be defined as in Equation~\eqref{eqn:F(x)} for each level-$l$ pseudo-utility.
\begin{assumption}\label{assumption:unique}
	$F^{(l)}$ satisfy the  $P_{\Upsilon}$ condition.
\end{assumption}
Note that these assumptions directly generalize the conditions required for the convergence of BRD to our multi-scale pseudo-utilities.
The following theorem formally states that \textsc{SH-BRD} converges to a NE for \emph{ 2-level games}.
\begin{theorem} \label{thm:sep_alg_conv}
Suppose $L=2$.
If Assumptions~\ref{assumption:cont_diff} and~\ref{assumption:unique} hold, \textsc{SH-BRD} converges to a NE, which is unique.
\end{theorem}
The full proof of this theorem, which makes use of the connection between \textsc{SH-BRD} and SVI, is provided in the Supplement due to space constraint.
The central issue, however, is that there are no established convergence guarantees for ADM-based algorithms for SVI with 3 or more separable operators.
Alternative algorithms for SVI can extend to the case of 3 operators using parallel operator updates with regularization terms, but no approaches exist that can handle more than 3 operators~\citep{He2009ParaALM}.
We thus propose an algorithm for iteratively solving multi-scale games that uses the general idea from \textsc{SH-BRD}, but packs all levels into two \emph{meta-levels}.
The two meta-levels each has to be comprised of consecutive levels.
For example, if we have 5 levels, we can have $\{1,2,3\}$ and $\{4,5\}$ combinations, but not $\{1,2,4\}$ and $\{3,5\}$.
Upon grouping levels together to obtain a meta-game with only two meta-levels, we can apply what amounts to a 2-level version of the  \textsc{SH-BRD}.
This yields an algorithm, which we call \emph{Hybrid Hierarchical BRD (HH-BRD)}, that now provably converges to a NE for an arbitrary number of levels $L$ given assumptions 1-3.


As presenting the general version of \textsc{HH-BRD} involves cumbersome notation, we illustrate the idea by presenting it for a 4-level game (Algorithm \ref{alg:hybrid}).
The fully general version is deferred to the Supplement.
In this example, 
the objectives of the meta-levels are defined as
\begin{eqnarray*}
	\hat{u}^{(sl_1)}_{i} &=& u^{(1)}_{i} + u^{(2)}_{k_{i2}} - L_{k_{i3}}^{(sl_1,sl_2)}\bigg( \sigma^{(3)}_{k_{i3}} ( \pmb{x}_{S^{(3)}_{k_{i3}}} ) , x^{(3)}_{k_{i3}} \bigg), \\
%
    \hat{u}^{(sl_2)}_{k_{i3}} &=& u^{(3)}_{k_{i3}} + u^{(4)}_{k_{i4}} - L_{k_{i3}}^{(sl_2,sl_1)}\bigg( x^{(3)}_{k_{i3}}, \sigma^{(3)}_{k_{i3}} ( \pmb{x}_{S^{(3)}_{k_{i3}}} )  \bigg)~. 
\end{eqnarray*}

\begin{algorithm}[ht]\label{alg:hybrid}
	\SetAlgoLined
	
	Initialize the game, $t=0, x^{(1)}_i(0)=(\pmb{x}_0)_i, i = 1,\dots,N^{(0)}$
	
	\For{l = 2:4}{
	    \For{k = 1:$N^{(l)}$}{
		    $\pmb{x}^{(l)}_k(0) = \sigma^{(l)}_k ( \pmb{x}_{S^{(l)}_k}(0) )$;\\
		}
	}
	\While{not converged}{
		\For{k = 1:$N^{(3)}$ (Meta-Level-1 Penalty Update)}{
			Update $L_k^{(sl_1, sl_2)}$
		}
		\For{$i = 1:N^{(1)}$ (Level-1)}{
			$x^{(1)}_i(t + 1) = BR_i \bigg(\pmb{x}^{(1)}_{I_i}(t), \pmb{x}^{(2)}_{I_{k_{i2}}}(t), x^{(3)}_{k_{i3}}(t), \hat{u}^{(sl_1)}_i \bigg)$\\
		}
		
		\For{j = 1:$N^{(2)}$ (Level-2)}{
			$\pmb{x}^{(2)}_j(t+1) = \sigma^{(2)}_j ( \pmb{x}_{S^{(2)}_j}(t+1) )$\\
		}
		
		\For{k = 1:$N^{(3)}$ (Meta-Level-2 Penalty Update)}{
			Update $L_k^{(sl_2, sl_1)}$
		}
		\For{$k = 1:N^{(3)}$ (Level-3)}{
		    \vspace{-0.1in}
			\begin{align*}
			    x^{(3)}_k(t + 1) = BR_i &\bigg(  \sigma^{(3)}_k( \pmb{x}_{S_k^{(3)}}(t+1) ), \pmb{x}^{(3)}_{I_k}(t), \\& \pmb{x}^{(4)}_{-p}(t),
			    \hat{u}^{(sl_2)}_k \bigg),  (a^{(3)}_k \in S_p^{(4)})
			\end{align*}
		}
		
		\For{p = 1:$N^{(4)}$ (Level-4)}{
			$\pmb{x}^{(4)}_p(t+1) = \sigma^{(4)}_p ( \pmb{x}_{S^{(4)}_p}(t+1) )$\\
		}
		
		$t \leftarrow t+1$;
	}
	\caption{Hybrid Hierarchical BRD} 
\end{algorithm}

\begin{theorem} \label{thm:hybrid_alg_conv}
	Suppose Assumptions~\ref{assumption:cont_diff}-\ref{assumption:unique} hold
	Then \textsc{HH-BRD} finds the unique NE.
\end{theorem}

\begin{proof}[Proof Sketch]
We first ``flatten'' the game within each meta-level to obtain an effective 2-level game.
We then use Theorem \ref{thm:sep_alg_conv} to show this 2-level game converges to the unique NE of the game under SH-BRD. 
Finally, we prove that \textsc{SH-BRD} and \textsc{HH-BRD} have the same trajectory given the same initialization, thus establishing the convergence for \textsc{HH-BRD}.
 For full proof see Supplement, Appendix \ref{appendix:proof_thm_hybrid}.
\end{proof}

HH-BRD combines the advantages of both MS-BRD and SH-BRD: not only does it exploit the sparsity embedded in the network topology, but it also avoids the convergence problem of SH-BRD when the number of levels is higher than three.
Indeed, there is a known challenge in the related work on structured variational inequalities that convergence is difficult when we involve three or more operators \citep{He2009ParaALM}, which we leverage for our convergence results, with operators mapping to levels in our multi-scale game representation.
One may be concerned that \textsc{HH-BRD} pseudocode appears to involve greater complexity (and more steps) than \text{SH-BRD}.
However, this does not imply greater algorithmic complexity, but is rather due to our greater elaboration of the steps within each super level.
Indeed, as our experiments below demonstrate, the superior theoretical convergence of \textsc{HH-BRD} also translates into a concrete computational advantage of this algorithm.


\section{Numerical Results and Analysis} \label{sec:numerical}

In this section, we numerically compare the three algorithms  introduced in Section \ref{sec:algorithms}, as well as the conventional BRD.
We only consider settings which satisfy Assumptions 1-3; consequently, we focus comparison on computational costs.
We use two measures of computational cost: floating-point operations (FLOPs) in the case of games with a linear best response (a typical measure for such settings), and CPU time for the rest.
All experiments were performed on a machine with A 6-core 2.60/4.50 GHz CPU with hyperthreaded cores, 12MB Cache, and 16GB RAM.

\smallskip
\noindent{\bf Games with a Linear Best Response (GLBRs) }
GLBRs \citep{bramoulle2014strategic, candogan2012optimal, IDS-Games} 
feature utility functions 
such that an agent's best response is a linear function of its neighbors' actions. This includes quadratic utilities of the form
\begin{equation}\label{eqn:glbr}
    u_i(x_i, x_{I_i}) = a_i + b_i x_i + \bigg(\sum_{j \in I_i} g_{ij} x_j\bigg) x_i - c_i x_i^2, 
\end{equation}
since an agent's best response is:
\begin{equation*}
    BR_i(x_{I_i}, u_i) = \frac{ \sum_{j \in I_i} g_{ij} x_j }{2 c_i} - b_i.
\end{equation*}

We consider a 2-level GLBR and compare three algorithms: \textsc{BRD} (baseline), \textsc{MS-BRD}, and \textsc{HS-BRD} (note that in 2-level games, \textsc{HH-BRD} is identical to \textsc{HS-BRD}, and we thus don't include it here).
We construct random 2-level games with utility functions based on Equation~\eqref{eqn:glbr}.
Specifically, we generalize this utility so that Equation~\eqref{eqn:glbr} represents only the level-1 portion, $u_i^{(1)}$, and let the level-2 utilities be
\[
u_k^{(2)}(x_k,\pmb{x}_{I_k}) = x_k^{(2)}\sum_{p \ne k} v_{kp}x_p^{(2)}
\]
for each group $k$. 
At every level, the existence of a link between two agents follows the Bernoulli distribution where $P_{exist} = 0.1$. If a link exists, we then generate a parameter for it.
The parameters of the utility functions are sampled uniformly in $[0,1]$ without requiring symmetry. Please refer to Appendix \ref{appendix:numerical_data} and \ref{appendix:linear_data} for further details.
Results comparing \textsc{BRD}, \textsc{MS-BRD}, and \textsc{SH-BRD} are shown in Table~\ref{table:num_flops_sparse}.
We observe dramatic improvement in the scalability of using \textsc{MS-BRD} compared to conventional \textsc{BRD}.
This improvement stems from the representational advantage provided by multi-scale games compared to conventional graphical games (since without the multi-scale representation, we have to use the standard version of \textsc{BRD} for equilibrium computation).
We see further improvement going from \textsc{MS-BRD} to \textsc{SH-BRD} which makes algorithmic use of the multi-scale representation.

\begin{table}[htbp]
	\begin{center}
		\begin{tabular}{p{0.02\textwidth} p{0.125\textwidth} p{0.125\textwidth} p{0.125\textwidth}}
			\toprule
			\textbf{Size} & \textbf{BRD} & \textbf{MS-BRD} & \textbf{SH-BRD}\\
			\midrule
			$30^2$ & (2.51$\pm$0.18)$\times 10^6$ & (1.03$\pm$0.07)$\times 10^5$
			& \textbf{(9.81$\pm$0.81)$\times 10^4$}
			\\
			
			\midrule
			$50^2$ & (2.53$\pm$0.18)$\times 10^7$ & (5.33$\pm$0.04)$\times 10^5$
			& \textbf{(4.35$\pm$0.07)$\times 10^5$}
			\\
			
			\midrule
			$100^2$ & (4.46$\pm$0.32)$\times 10^8$ & (4.36$\pm$0.31)$\times 10^6$
			& \textbf{(3.56$\pm$0.29)$\times 10^6$}
			\\
			
			\midrule
			$200^2$ & (6.73$\pm$0.58)$\times 10^9$ & (3.48$\pm$0.29)$\times 10^7$
			& \textbf{(2.79$\pm$0.21)$\times 10^7$}
			\\
			
			\midrule
			$500^2$ & (2.84$\pm$0.21)$\times 10^{11}$ & (5.69$\pm$0.41)$\times 10^8$
			& \textbf{(4.04$\pm$0.29)$\times 10^8$}
			\\
			
			\bottomrule
		\end{tabular}
	\end{center}
	
	\caption{Convergence and complexity (flops) comparison  with linear best response under multiple initialization.}\label{table:num_flops_sparse}
\end{table}

\smallskip
\noindent{\bf Games with a Non-Linear Best Response }
Next, we study the performance of the proposed algorithms in 2- and 3-level games, with the same number of groups in each level (we systematically vary the number of groups).
Since \textsc{SH-BRD} and \textsc{HH-BRD} are identical in 2-level games, the latter is only used in 3-level games.
All results are averaged over 30 generated sample games.
The non-linear best response fits a much broader class of utility functions than the linear best response. The best responses generally don't have closed-form representations. In this case, we can't use linear equations to find the best response and instead have to apply gradient-based methods. 
In our instances, the utility with non-linear best responses is generated by adding an exponential cost term to the utility function used in GLBRs. Please refer to Appendix \ref{appendix:numerical_data} and \ref{appendix:nonlinear_data} for further details.

\begin{table}[htbp]
	\begin{center}
		\begin{tabular}{p{0.05\textwidth} p{0.11\textwidth} p{0.11\textwidth} p{0.11\textwidth}}
			
			\toprule
			\textbf{Size} & \textbf{BRD} & \textbf{MS-BRD} & \textbf{SH-BRD}\\
			\midrule
			$30^2$ & 1.50$\pm$0.05 & 1.02$\pm$0.02 & \textbf{0.54$\pm$0.01}
			\\
			
			\midrule
			$50^2$ & 26.70$\pm$0.36 & 3.70$\pm$0.14
			& \textbf{1.81$\pm$0.04}
			\\
			
			\midrule
			$100^2$ & 1512$\pm$9 & 23.81$\pm$0.69
			& \textbf{12.10$\pm$0.13}
			\\
			
			\midrule
			$200^2$ & $>18000$ & 287.2$\pm$5.4 & \textbf{133.6$\pm$2.5}
			\\
			
			\midrule
			$500^2$ & nan & 5485$\pm$13 & \textbf{2524$\pm$10}
			\\
			
			\bottomrule
		\end{tabular}
	\end{center}
	
	\caption{CPU times on a single machine on 2-Level games with general best response functions; all times are in seconds.}\label{table:L2_CPU_sparse}
\end{table}

Table \ref{table:L2_CPU_sparse} shows the CPU time comparison between all algorithms.
The scalability improvements from our proposed algorithms are substantial, with orders of magnitude speedup in some cases (e.g., from $\sim$ 25 minutes for the \textsc{BRD} baseline, down to $\sim$ 12 seconds for \textsc{SH-BRD} for games with 10K agents).
Furthermore, \textsc{BRD} fails to solve instances with 250K agents, which can be solved by \textsc{SH-BRD} in $\sim$ 42 min.
Again, we separate here the representational advantage of multi-scale games, illustrated by \textsc{MS-BRD}, and algorithmic advantage that comes from \textsc{SH-BRD}.
Note that \textsc{SH-BRD}, which takes full advantage of the multi-scale structure, also exhibits significant improvement over \textsc{MS-BRD}, yielding a factor of 2-3 reduction in runtime.

%

			
			
			
			
			
			
	

\begin{table}[htbp]
	\begin{center}
		\begin{tabular}{p{0.03\textwidth} p{0.11\textwidth} p{0.11\textwidth} p{0.12\textwidth}}
			
			\toprule
			\textbf{Size} & \textbf{BRD} & \textbf{MS-BRD} & \textbf{SH-BRD}\\
			\midrule
			$30^2$ & 1.21$\pm$0.04 & 0.63$\pm$ 0.01 & \textbf{0.037$\pm$0.003}
			\\
			
			\midrule
			$50^2$ & 23.88$\pm$0.16 & 1.99$\pm$0.04
			& \textbf{0.079$\pm$0.004}
			\\
			
			\midrule
			$100^2$ & 1461$\pm$14 & 15.49$\pm$0.24
			& \textbf{0.304$\pm$0.006}
			\\
			
			\midrule
			$200^2$ & $>18000$ & 192.0$\pm$1.2 & \textbf{1.87$\pm$0.05}
			\\
			
			\midrule
			$500^2$ & nan & 4258$\pm$56 s & \textbf{28.79$\pm$0.37}
			\\
			
			\bottomrule
		\end{tabular}
	\end{center}
	
	\caption{CPU times on a single machine for 2-Level, linear/nonlinear best-response games; all times are in seconds.}\label{table:Lin+Nonlin_CPU_sparse}
\end{table}
Our next set of experiments involves games in which level-1 utility has a linear best response, but level-2 utility has a non-linear best response.
The results are shown in Table~\ref{table:Lin+Nonlin_CPU_sparse}.
We see an even bigger advantage of \textsc{SH-BRD} over the others: it is now typically orders of magnitude faster than even \textsc{MS-BRD}, which is itself an order of magnitude faster than \textsc{BRD}.
For example, in games with 250K agents, in which \textsc{BRD} fails to return a solution, \textsc{MS-BRD} takes more than 1 hour to find a solution, whereas \textsc{SH-BRD} finds a solution in under 30 seconds.


\begin{table}[htbp]
	\begin{center}
		\begin{tabular}{p{0.02\textwidth} p{0.085\textwidth} p{0.085\textwidth} p{0.085\textwidth} p{0.085\textwidth}}
			
			\toprule
			\textbf{Size} 
			& \textbf{BRD} & \textbf{MS-BRD} & \textbf{SH-BRD} & \textbf{HH-BRD}\\
			\midrule
			$10^3$ & 1.23$\pm$0.03 & 0.59$\pm$0.01 & 0.76$\pm$0.03 & \textbf{0.43$\pm$0.02}
			\\
			
			\midrule
			$20^3$ & 696.0$\pm$8.7 & 3.78$\pm$0.09 & 6.05$\pm$0.08 & \textbf{3.35$\pm$0.09}
			\\
			
			\midrule
			$30^3$ & $>$ 18000 & 15.70$\pm$0.11 & 25.13$\pm$0.14
			& \textbf{13.39$\pm$0.11}
			\\
			
			\midrule
			$50^3$ & nan & 68.59$\pm$0.75 & 138.8$\pm$1.1
			& \textbf{57.98$\pm$0.69}
			\\
			
			\midrule
			$100^3$ & nan & 1126$\pm$6 & 2343$\pm$21
			& \textbf{877.1$\pm$11.5}
			\\
			
			\bottomrule
		\end{tabular}
	\end{center}
	
	\caption{CPU times in seconds on a single machine on 3-Level, general best response games; all times are in seconds.}	\label{table:L3_CPU_sparse}
\end{table}
Finally, Table~\ref{table:L3_CPU_sparse} presents the results of 
\textsc{HH-BRD} in games with $>2$ levels compared to \textsc{SH-BRD}, which does not provably converge in such games. In this case, \textsc{HH-BRD} outperforms the other  alternatives, with up to 22\% improvement over \textsc{MS-BRD}; indeed, we find that \textsc{SH-BRD} is considerably worse even than \textsc{MS-BRD}.

\section{Conclusions and Future Directions} \label{sec:conclusion}

We proposed a novel representation of games that have a multi-scale network structure.
These generalize network games, but with special structures that agent utilities are additive across the levels of hierarchy, with utility at each level depending only on the aggregate strategies of other groups.
We present several iterative algorithms that make use of the multi-scale game structure, and show that they converge to a pure strategy Nash equilibrium under similar conditions as  for best response dynamics in network games.
Our experiments demonstrate that the proposed algorithms can yield orders of magnitude scalability improvement over conventional best response dynamics. 
Our multi-scale algorithms can reveal to what extent one’s group affiliation impacts one’s strategic decision making, and how strategic interactions among groups impact strategic interactions among individuals. 

While the issue of multi-scale networks abounds in the network science literature (e.g., hierarchical clustering, etc.), the “multi-scale” part is primarily concerned with community structure in networks, rather than modeling how \textit{how communities interact}, which is critical for us in describing a formal multi-scale structure for games.  Thus 
a very important future direction is to identify and obtain relevant field data for experiments, and create realistic benchmarks for multi-scale games.  This would involve identifying ways to obtain data about how communities (and not just individuals) interact. 
%
Once we have the ability to collect data about interactions at multiple scales (e.g., among members and among groups), 
 we can apply our algorithms to such multi-scale networks.
 To use criminal networks (criminal organizations and their members) as an example, given game models constructed with the help of domain expertise, we can:
\begin{enumerate}
    \item compute equilibria predicting, say, criminal activity as a function of structural changes to organizations;
    \item infer utility models from observational data at multiple scales;
    \item study policies (including strengthening or weakening connections between agents or groups, endowing agents/groups with more resources (lower costs of effort), etc.) that would induce more desirable equilibrium outcomes.
\end{enumerate}

\section*{Acknowledgment}
This work is supported by the NSF under grants CNS-1939006, CNS-2012001, IIS-1905558 (CAREER) and by the ARO under contract W911NF1810208. 

\bibliographystyle{aaai}
\bibliography{papers}

\appendix

\clearpage
\begin{center}
    \textbf{\Large{Appendices}}    
\end{center}

\section{Structured Variational Inequalities}

%

A structured variational inequality SVI$_n$ arises when a VI problem has $n$ separable operators. This is used to analyze our game under the multi-scale perspective described in Section \ref{sec:model}. 

%
We now introduce a particular type of SVI$_2$ relevant to our model. 
Suppose the $N$ level-1 agents form $M$ disjoint groups in the game and $S_j$ denotes the $j$th level-1 group, whereby $i \in S_j$ denotes that $a_i$ is a member of $S_j$. Consider the following utility function of $a_i$: 
\begin{equation} \label{eqn:SVI2_utility}
    u_i(x_i, \pmb{x}_{-i}, y_j, \pmb{y}_{-j}) = u^{(1)}_i(x_i, \pmb{x}_{-i}) + u^{(2)}_j(y_j, \pmb{y}_{-j}),
\end{equation}
where $\pmb{x} \in \mathbf{R}^N$ denotes the level-1 action profile and $\pmb{y} \in \mathbf{R}^M$ denotes the level-2 action profile, and $A\pmb{x} + \pmb{y} = \pmb{0}$, for
\begin{equation*}
A_{ji} = \bigg \{ \begin{aligned}
           & -1, & \text{if } i \in S_j\\
                    & 0, & \text{else}
                \end{aligned} ~~, j = 1,\dots,M, i = 1,\dots,N~. 
\end{equation*}
Thus $A\pmb{x} + \pmb{y} = \pmb{0}$ is equivalent to $y_j = \sum_{i \in S_j} x_i$. We say $\pmb{x}$ and $\pmb{y}$ are two separated operators, and define
\begin{align} \label{eqn:SVI2_formation}
    & F^{(1)}(\pmb{x}) := \bigg( - \triangledown_{x_i} u^{(1)}_i(\pmb{x}) \bigg)_{i=1}^{N}, ~~x_i \in K^{(1)}_i, \nonumber \\
    & F^{(2)}(\pmb{y}) := \bigg( - \triangledown_{y_j} u^{(2)}_j(\pmb{y}) \bigg)_{j=1}^{M}, ~~y_j \in K^{(2)}_j,  \nonumber \\
    & K^{(1)} = \prod_{i=1}^N K^{(1)}_i, K^{(2)} = \prod_{j=1}^M K^{(2)}_j, \overline{K} = K^{(1)} \times K^{(2)},\nonumber \\
    & \pmb{v} = \begin{bmatrix}
        \pmb{x} \\ \pmb{y}
        \end{bmatrix} \in \overline{K},~~ 
    \overline{F}(\pmb{v}) = \begin{bmatrix} F^{(1)}(\pmb{x}) \\ F^{(2)}(\pmb{y}) \end{bmatrix}.
\end{align}

Define $\Omega = \{v \in \overline{K} | A\pmb{x} + \pmb{y} = \pmb{0} \}$. Then the VI($\Omega, \overline{F}$) problem is to find $v^* \in \Omega$, such that: 
$	(\pmb{v} - \pmb{v}^*)^T \overline{F}(\pmb{v}) \geq 0, ~ \forall \pmb{v} \in \Omega.$
This problem is equivalent to the SVI$_2$ problem VI($\mathcal{W}, Q$) defined in Eqn (\ref{eqn:SVI_problem})
\begin{equation} \label{eqn:SVI_problem}
	(\pmb{\omega} - \pmb{\omega}^*)^T Q(\pmb{\omega}) \geq 0, ~ \forall \pmb{\omega} \in \mathcal{W},
\end{equation}
where, $\mathcal{W} = \overline{K} \times \mathbf{R}^M$ and
\begin{equation} \label{eqn:SVI_notation}
    \pmb{\omega} = \begin{pmatrix}
		\pmb{x} \\ \pmb{y} \\ \pmb{\lambda}
		\end{pmatrix}, 
		Q(\pmb{\omega}) = \begin{pmatrix}
		F^{(1)}(\pmb{x}) - A^T \pmb{\lambda} \\ F^{(2)}(\pmb{y}) - \pmb{\lambda} \\ A \pmb{x} + \pmb{y} \end{pmatrix}.
\end{equation}
%
It is easy to see that if we use $\sum_{i \in S_j} x_i$ to replace $y_j$, then we again have a single operator $\pmb{x}$ and can construct a VI($K,F$) as outlined in Section \ref{sec:preliminaries}. There is a one-to-one mapping between a solution $\pmb{x}^*$ to VI($K,F$) and a solution $\pmb{\omega}^* = (\pmb{x}^*, - A\pmb{x}^*, \pmb{\lambda}^*)$ to VI($\mathcal{W}, Q$). Therefore, solving either VI($K,F$) or VI($\mathcal{W}, Q$) finds the set of NEs.

\section{Uniqueness of NE}

We will introduce some special matrices before we move on to the sufficient conditions for the uniqueness of NE.

\begin{definition} \label{defn:matrices}
	Some special matrices:
	\begin{enumerate}
		\item P-matrix: A square matrix is a P-matrix if all its principal components have positive determinant
		\item Z-matrix: A square matrix is a Z-matrix if all its off-diagonal components are nonpositive
		\item M-matrix: An M-matrix is a Z-matrix whose eigenvalues' real parts are nonnegative
		\item L-matrix: An L-matrix is a Z-matrix whose diagonal elements are nonnegative
	\end{enumerate}
\end{definition}

For an arbitrary mapping $F: \mathbf{R}^N \rightarrow \mathbf{R}^N$, we denote the Jacobian of $F(\pmb{x})$ as $JF(\pmb{x})$. And then $\triangledown_j F_i = [JF(\pmb{x})]_{ij}$

Checking if a matrix is P-matrix or not is still not trivial, and we can look at the spectral radius of a matrix instead.

\begin{theorem} \label{thm:P_gamma}
	\textbf{The $P_{\Gamma}$ condition}:
	
	We define the $\Gamma$ matrix generated from $F$ as follows
	\begin{equation} \label{eqn:gamma}
	\Gamma(F) = \begin{bmatrix}
	0 & -\frac{\beta_{1,2}(F)}{\alpha_1(F)} & \cdots & -\frac{\beta_{1,N}(F)}{\alpha_1(F)} \\
	-\frac{\beta_{2,1}(F)}{\alpha_2(F)} & 0 & \cdots & -\frac{\beta_{2,N}(F)}{\alpha_2(F)} \\
	\vdots & \vdots & \ddots & \vdots \\
	-\frac{\beta_{N,1}(F)}{\alpha_N(F)} & -\frac{\beta_{N,2}(F)}{\alpha_N(F)} & \cdots & 0
	\end{bmatrix},
	\end{equation}
	if the spectral radius $\rho(\Gamma(F)) = ||\Gamma(F)||_2 < 1$, then we say $F$ satisfies the $P_\Gamma$ condition. Then $P_{\Gamma}$ condition $\Leftrightarrow$ $P_{\Upsilon}$ condition and VI$(K,F)$ has a unique solution.
\end{theorem}

In \cite{scutari2014real}, the authors mentioned that the $P_{\Upsilon}$ captures ``some kind of diagonal dominance''. In fact, the strong diagonal dominance(s.d.d) or weakly chained diagonal dominance(w.c.d.d) of $\Upsilon$ can be an easier yet sufficient condition to check.

\begin{theorem} \label{thm:PM2DD}
	If $\Upsilon$ is s.d.d or w.c.d.d, the NE is unique, since
	\begin{align*}
	&~ \text{$\Upsilon$ is an s.d.d L-matrix}\\
	\Rightarrow &~ \text{$\Upsilon$ is a w.c.d.d L-matrix}\\
	\Leftrightarrow &~ \text{$\Upsilon$ is a nonsigular weakly diagonally dominant(w.d.d)}\\
	&~ \text{L-matrix}\\
	\Leftrightarrow &~ \text{$\Upsilon$ is a nonsigular w.d.d M-matrix}\\
	\Rightarrow &~ \text{$\Upsilon$ is a P-matrix}
	\end{align*}
\end{theorem}

Also, when $\Upsilon$ is s.d.d, $\Gamma$ is a (right, row) substochastic matrix and thus $\rho(\Gamma) < 1$ trivially holds and the NE is unique.

The $P_{\Upsilon}$ condition guarantees both the uniqueness of NE and the convergence of BRD. Please refer to \cite{parise2019variational} for more conditions on the uniqueness.

\section{Proof of Theorem \ref{thm:sep_alg_conv}} \label{appendix:proof_thm_sep}
\begin{proof}
This algorithm is designed to solve the SVI problem presented in Eqn (\ref{eqn:SVI_problem}) and (\ref{eqn:SVI_notation}).
%
	We denote $H = \frac{1}{2}\textbf{diag}(\pmb{h})$, and the norm $||\pmb{x}||_{G}$, where $G \succ \pmb{0}$ as
	\begin{equation*}
		||\pmb{x}||_{G} = \pmb{x}^T G \pmb{x}.
	\end{equation*}
    
    For simplicity reason, we will use $\pmb{x}$ and $\pmb{y}$ to replace $\pmb{x}^{(1)}$ and $\pmb{x}^{(2)}$ in the remainder of the proof.
    
	We can rewrite the steps in Algorithm \ref{alg:separated} as follows:
	\begin{itemize}
		
	\item Step 0: Initialization, given $\epsilon, \mu$ and $\pmb{x}_0$, let $t=0$, $\pmb{x}(0) = \pmb{x}_0$, $y_k(0) = \sigma_k(\pmb{x}_{S_k}(0))$; arbitrarily choose $\pmb{\lambda}(0)$.
	
	\item Step 1: Find $\pmb{x}^* \in K^{(1)}$ that solves
	\begin{equation} \label{eqn:level1_diff}
	(\pmb{x}' - \pmb{x}^*)^T \bigg[ f(\pmb{x}^*) - A^T[\pmb{\lambda}(t) - H(A\pmb{x}^*+ \pmb{y}(t))] \bigg] \geq 0,
	\end{equation}
	for $\forall \pmb{x}' \in K$, and set $\pmb{x}(t+1) = \pmb{x}^*$.
	
	\item Step 2: Find $\pmb{y}^* \in K^{(2)}$ that solves
	\begin{align} \label{eqn:level2_diff}
	(\pmb{y}' - \pmb{y}^*)^T \bigg[ f(\pmb{x}^*) - [\pmb{\lambda}(t) - H(A\pmb{x}(t+1)+\pmb{y}^*)] \bigg] \geq 0,
	\end{align}
	for $\forall \pmb{y}' \in K$, and set $\pmb{y}(t+1) = \pmb{y}^*$.
	
	\item Step 3: Set
	\begin{equation} \label{eqn:lambda_update_app}
	\pmb{\lambda}(t+1) = \pmb{\lambda}(t) - H(A\pmb{x}(t+1) - \pmb{y}(t+1))
	\end{equation}
	
	\item Step 4: Convergence verification: If $||\pmb{\omega}(t+1) - \pmb{\omega}(t)||_{\infty} < \epsilon$, then stop. Otherwise let $t \leftarrow t+1$ and go back to Step 1.
	
	\end{itemize}
	
	When we have $\pmb{y}(t+1) = \pmb{y}(t)$ and $\pmb{\lambda}(t+1) = \pmb{\lambda}(t)$, $\pmb{\omega}(t+1) = (\pmb{x}(t+1), \pmb{y}(t+1), \pmb{\lambda}(t+1))$ is the solution to our SVI$_2$. We denote the unique solution as $\pmb{\omega}^* = (\pmb{x}^*, \pmb{y}^*, \pmb{\lambda}^*)$. From Eqn (\ref{eqn:level2_diff}) and (\ref{eqn:lambda_update_app}), we have the following from Section 2 of \citep{He2009ParaALM},
	\begin{align} \label{eqn:alg3_contracttion}
		&~ ||\pmb{y}(t+1) - \pmb{y}^*||^2_H + ||\pmb{\lambda}(t+1) - \pmb{\lambda}^*||^2_{H^{-1}} \nonumber \\
		\leq &~ \bigg(||\pmb{y}(t) - \pmb{y}^*||^2_H + ||\pmb{\lambda}(t) - \pmb{\lambda}^*||^2_{H^{-1}} \bigg) \nonumber \\
		&~ - \bigg(||\pmb{y}(t+1) - \pmb{y}(t)||^2_H + ||\pmb{\lambda}(t+1) - \pmb{\lambda}(t)||^2_{H^{-1}} \bigg) \nonumber \\ 
		< &~ ||\pmb{y}(t) - \pmb{y}^*||^2_H + ||\pmb{\lambda}(t) - \pmb{\lambda}^*||^2_{H^{-1}},
	\end{align}
	which shows the contraction property of the sequence $\{(\pmb{y}(t), \pmb{\lambda}(t) ) \}$ and thus proves the convergence of the algorithm.
	
	A more detailed proof of convergence of the above steps in Eqn (\ref{eqn:level1_diff})-(\ref{eqn:lambda_update_app}) is covered in \citep{Gabay1976, Glowinski1985}, and a more generalized version of the above steps and convergence proofs are covered in \citep{Tseng1990, Lions1979}.
	
\end{proof}

\section{Proof of Theorem \ref{thm:hybrid_alg_conv}} \label{appendix:proof_thm_hybrid}

\subsection{Full version of HH-BRD}
We will first show the ull version of HH-BRD, suppose the superlevel partitions is taken between level $q-1$ and level $q$, then for $i = 1,\dots,N^{(1)}$,
\begin{align} \label{eqn:HH_BRD_util1}
	\hat{u}^{(sl_1)}_i = & \sum_{l=1}^{q-1} u^{(l)}_{k_{il}} (x_{k_{il}}, \pmb{x}_{I_{k_{il}}}) \nonumber \\
	& - L_{k_{iq}}^{(sl_1,sl_2)}\bigg( \sigma^{(1,q)}_{k_{iq}} ( \pmb{x}_{S^{(1,q)}_{k_{iq}}} ) , x^{(q)}_{k_{iq}} \bigg),
\end{align}
where
\begin{equation*}
    S^{(1,q)}_p = \{ a^{(1)}_i ~|~ k_{iq} = p \}, 
\end{equation*}
\begin{equation*}
    \sigma^{(1,q)}_p ( \pmb{x}_{S^{(1,q)}_p} ) = \sum_{a^{(1)}_i \in S^{(1,q)}_p} x^{(1)}_i.
\end{equation*}

And for $j = 1,\dots,N^{(q)}$
\begin{align} \label{eqn:HH_BRD_util2}
    \hat{u}^{(sl_2)}_j = & \sum_{l=q}^{L} u^{(l)}_{k_{jl}} (x_{k_{jl}}, \pmb{x}_{I_{k_{jl}}}) \nonumber \\
    & - L_j^{(sl_2,sl_1)}\bigg( x^{(q)}_j, \sigma^{(1,q)}_j ( \pmb{x}_{S^{(1,q)}_j} )  \bigg)~. 
\end{align}

Please refer to Algorithm \ref{alg:hybrid_full} for the pseudo code of the full version of this algorithm. The loss function updates are similar to that of Algorithm \ref{alg:separated}.
\begin{algorithm}[h]\label{alg:hybrid_full}
	\SetAlgoLined
	
	Initialize the game, $t=0, x^{(1)}_i(0)=(\pmb{x}_0)_i, i = 1,\dots,N^{(0)}$
	
	\For{l = 2:L}{
	    \For{k = 1:$N^{(l)}$}{
		    $\pmb{x}^{(l)}_k(0) = \sigma^{(l)}_k ( \pmb{x}_{S^{(l)}_k}(0) )$;\\
		}
	}
	\While{not converged}{
		\For{k = 1:$N^{(q)}$ (Meta-Level-1 Penalty Update)}{
			Update $L_k^{(sl_1, sl_2)}$
		}
		\For{$i = 1:N^{(1)}$ (Level-1/Meta-Level-1 Gaming)}{
			$x^{(1)}_i(t + 1) = BR_i \bigg(\pmb{x}^{(1)}_{I_i}(t), \pmb{x}^{(2)}_{I_{k_{i2}}}(t), \dots, x^{(3)}_{k_{iq}}(t), \hat{u}^{(sl_1)}_i \bigg)$\\
		}
		\For{l = 2:q-1 (Level-2 to Level-q Aggregation)}{
    		\For{j = 1:$N^{(l)}$}{
    			$\pmb{x}^{(l)}_j(t+1) = \sigma^{(l)}_j ( \pmb{x}_{S^{(l)}_j}(t+1) )$\\
    		}
		}
		\For{k = 1:$N^{(q)}$ (Meta-Level-2 Penalty Update)}{
			Update $L_k^{(sl_2, sl_1)}$
		}
		
		\For{$j = 1:N^{(q)}$ (Level-q/Meta-Level-2 Gaming)}{
		\begin{align*}
		    &~ x^{(q)}_j(t + 1) \\
		    = &~ BR_j \bigg( \sigma^{(1,q)}_j( \pmb{x}_{S_j^{(1,q)}} ), \pmb{x}^{(q)}_{I_j}(t), \pmb{x}^{(q+1)}_{I_{k_{j(q+1)}} }(t), \dots,\\
		    &~~~~~~~~~~~ \pmb{x}^{(L)}_{I_{k_{jL}}}(t), \hat{u}^{(sl_2)}_j \bigg)
		\end{align*}
		}
		
		\For{l = q+1:L (Level-2 to Level-q+1 Aggregation)}{
    		\For{p = 1:$N^{(l)}$}{
    			$\pmb{x}^{(l)}_p(t+1) = \sigma^{(l)}_p ( \pmb{x}_{S^{(l)}_p}(t+1) )$\\
    		}
		}
		
		$t \leftarrow t+1$;
	}
	\caption{Hybrid Hierarchical BRD(Full Version)} 
\end{algorithm}

\subsection{Proof of Theorem}
We will first construct an equivalent 2-level game to the $L$-level game where $L>2$, and then show that the action profile update trajectories are the same for the original game and he equivalent game. Finally, the convergence of the equivalent game follows Theorem \ref{thm:sep_alg_conv} and thus Algorithm \ref{alg:hybrid} guarantees convergence.

\begin{proof}
    
	
	
	We define the following counter-part for utility component $u^{(l)}_i(x_i^{(l)}, \pmb{x}^{(l)}_{I_i} )$  ($1 < l < q$)
	\begin{equation} \label{eqn:util_comp_sl1}
	    \underline{u}^{(l)}_i( \pmb{x}_{S_i^{(1,l)}}, \pmb{x}_{S_{I_i}^{(1,l)}} ) = u^{(l)}_i(x_i^{(l)}, \pmb{x}^{(l)}_{I_i} ),
	\end{equation}
	when $x_i^{(l)} = \sigma^{(1,l)}_i ( \pmb{x}_{S^{(1,l)}_i} ), \forall i, \forall l \in \{2, \dots, q-1\}$. Both $\pmb{x}_{S_i^{(1,l)}}$ and $\pmb{x}_{S_{I_i}^{(1,l)}}$ are level-1 action profiles. This is exactly how we create the utility functions under the flat perspective, where we expand the higher level aggregate actions down to level-1.
	
	Similarly, we define the following counter-part for utility component $u^{(l)}_j(x_j^{(l)}, \pmb{x}^{(l)}_{I_j} )$  ($q < l \leq L$)
	\begin{equation} \label{eqn:util_comp_sl2}
	    \underline{u}^{(l)}_j( \pmb{x}_{S_j^{(q,l)}}, \pmb{x}_{S_{I_j}^{(q,l)}} ) = u^{(l)}_j(x_j^{(l)}, \pmb{x}^{(l)}_{I_j} ),
	\end{equation}
	when $x_j^{(l)} = \sigma^{(q,l)}_j ( \pmb{x}_{S^{(q,l)}_j} ), \forall j, \forall l \in \{q, \dots, L\}$. Both $\pmb{x}_{S_i^{(1,l)}}$ and $\pmb{x}_{S_{I_i}^{(1,l)}}$ are level-$q$ action profiles. This time we expand the higher level aggregate actions down to level-$q$ instead of level-1.
	
	So then we can define a ``flattened'' super-level-1 utility function counterpart for $u^{(sl_1)}_i$ as follows
	\begin{align} \label{eqn:util_sl1_alt}
	    \overline{u}^{(sl_1)}_i (x^{(1)}_i, x^{(1)}_{\underline{I}_i} ) = & \sum_{l=1}^{q-1} \underline{u}^{(l)}_{k_{il}}( \pmb{x}_{S_{k_{il}}^{(1,l)}}, \pmb{x}_{S_{I_{k_{il}}}^{(1,l)}} ) \nonumber \\
	    & - L_{k_{iq}}^{(sl_1,sl_2)}\bigg( \sigma^{(1,q)}_{k_{iq}} ( \pmb{x}_{S^{(1,q)}_{k_{iq}}} ) , x^{(q)}_{k_{iq}} \bigg),
	\end{align}
	where
	\begin{equation*}
	    \underline{I}^{(sl_1)}_i = \{ a^{(1)}_j | k_{jq} = k_{iq}, j \neq i \}.
	\end{equation*}
	
	Similarly, for meta-level 2, we can define a ``flattened''(to level-q) function counterpart for $u^{(sl_2)}_j$ as follows 
	\begin{align} \label{eqn:util_sl2_alt}
	    \overline{u}^{(sl_2)}_j (x^{(q)}_j, x^{(q)}_{\underline{I}^{(sl_2)}_j} ) = & \sum_{l=q}^{L} \underline{u}^{(l)}_{k_{jl}}( \pmb{x}_{S_{k_{jl}}^{(q,l)}}, \pmb{x}_{S_{I_{k_{jL}}}^{(q,l)}} )\nonumber \\
        & - L_j^{(sl_2,sl_1)}\bigg( x^{(q)}_j, \sigma^{(1,q)}_j ( \pmb{x}_{S^{(1,q)}_j} )  \bigg),
	\end{align}
	where
	\begin{equation*}
	    \underline{I}^{(sl_2)}_j = \{ a^{(q)}_p | k_{pL} = k_{jL}, p \neq j \}.
	\end{equation*}
	
	So now we can create a 2-level game where the level-1(resp. level-q) agents in the original game become the level-1(resp. level-2) agents in the new game with utility functions defined in Eqn (\ref{eqn:util_sl1_alt}) (resp. Eqn (\ref{eqn:util_sl2_alt})). Based on Theorem \ref{thm:sep_alg_conv}, we know that if we apply SH-BRD, we can converge to the unique NE of the game under Assumptions 1-3.
	
	Then it remains to show that given the same initialization, applying HH-BRD in the original game and the MS-BRD in the new 2-level game generate the same level-1 action profile update trajectory. This can be shown using induction.
	
	We know from initialization that
	\begin{equation*}
	    x_i^{(l)}(0) = \sigma^{(1,l)}_i ( \pmb{x}_{S^{(1,l)}_i} (0) ), \forall i, \forall l \in \{2, \dots, q-1\},
	\end{equation*}
	\begin{equation*}
	    x_j^{(l)}(0) = \sigma^{(q,l)}_j ( \pmb{x}_{S^{(q,l)}_j} (0) ), \forall j, \forall l \in \{q, \dots, L\}.
	\end{equation*}
	Then based on Eqn (\ref{eqn:util_comp_sl1}), we know that
	\begin{align*}
	     &~ u^{(sl_1)}_i(x^{(1)}_i, \pmb{x}^{(1)}_{I_i}(0), \dots, x^{(q)}_{k_{iq}}(0)) \\
	    = &~ \overline{u}^{(sl_1)}_i (x^{(1)}_i, x^{(1)}_{\underline{I}_i}(0) ) \\
	    \Leftrightarrow ~&~ BR_i(\pmb{x}^{(1)}_{I_i}(0), \dots, x^{(q)}_{k_{iq}}(0), u^{(sl_1)}_i) \\
	    = &~ BR_i(x^{(1)}_{\underline{I}_i}(0), \overline{u}^{(sl_1)}_i),
	\end{align*}
	and thus when $t=1$, $\pmb{x}^{(1)}(t)$ are the same when applying HH-BRD in the original game and the MS-BRD in the new 2-level game. Similarly, $\pmb{x}^{(q)}(1)$ are the same based on Eqn (\ref{eqn:util_comp_sl2}).
	
	Suppose $\pmb{x}^{(1)}(t)$ and $\pmb{x}^{(q)}(t)$ are the same for the two dynamics for $t = 0, 1, \dots, T$, we need to show that $\pmb{x}^{(1)}(t)$ and $\pmb{x}^{(q)}(t)$ are the same for $t = T+1$ to complete the proof.
	
	Again, based on Eqn (\ref{eqn:util_comp_sl1}), we know that
	\begin{align*}
	    &~ u^{(sl_1)}_i(x^{(1)}_i, \pmb{x}^{(1)}_{I_i}(T), \dots, x^{(q)}_{k_{iq}}(T)) \\
	    = &~ \overline{u}^{(sl_1)}_i (x^{(1)}_i, x^{(1)}_{\underline{I}_i}(T) ) \\
	    \Leftrightarrow ~&~ BR_i(\pmb{x}^{(1)}_{I_i}(T), \dots, x^{(q)}_{k_{iq}}(T), u^{(sl_1)}_i) \\
	    = &~ BR_i(x^{(1)}_{\underline{I}_i}(T), \overline{u}^{(sl_1)}_i),
	\end{align*}
	which implies $\pmb{x}^{(1)}(T+1)$ are the same for the two dynamics and similarly 	$\pmb{x}^{(q)}(T+1)$ are the same based on Eqn (\ref{eqn:util_comp_sl2}).
	
\end{proof}



\section{Data Generation for Numerical Experiments} \label{appendix:numerical_data}
We introduce the data generation procedures for both games with linear best response and non-linear best response in this part.

First of all, for both type of games, we create an adjacency matrix for each of the groups on every level. This matrix has 0 diagonal elements and for the off-diagonal elements, the existence of a directed edge subjects to the Bernoulli distribution where there is a fixed $P_{exist}$. Then if a directed edge exist, the edge weight is generated by choosing a value from $[0,1]$ uniformly at random. Later, we will multiply these matrices with different scalars to adjust the values so that Assumption 3 holds. These matrices have 0 diagonal elements because they capture the dependencies of agents on each other, or equivalently, they are used to model the external impact the agents receive from the network. The internal impact are modeled by cost functions and marginal benefit terms that only depend on an agent's own action.

\subsection{Linear Best Response Games} \label{appendix:linear_data}
For games with linear best response, we generated a 2-level game with 100 groups and 10,000 level-1 agents. The adjacency matrix generation follows $P_{exist}=0.1$, which creates a rather sparse network. Each level-2 group $S^{(2)}_k$ contains 100 members, and we use $W_k$ to denote the corresponding adjacency matrix. We use $V$ to denote the level-2 adjacency matrix. From Eqn (\ref{eqn:utility}), we know that for each level-1 agent, the utility function is
\begin{equation*}
    u_i(x^{(1)}_i, \pmb{x}^{(1)}_{I_i}, \pmb{x}^{(2)}_{I_{k_{i2}}}) = u^{(1)}_i(x^{(1)}_i, \pmb{x}^{(1)}_{I_i}) + u^{(2)}_{k_{i2}} ( x^{(2)}_{k_{i2}}, \pmb{x}^{(2)}_{I_{k_{i2}}} ), 
\end{equation*}
where
\begin{align*}
    u^{(1)}_i (x^{(1)}_i, \pmb{x}^{(1)}_{I_i}) = &~ b_i x^{(1)}_i + x^{(1)}_i \bigg( \sum_{j \in I_i} (W_{k_{i2}})_{r_i r_j} x^{(1)}_j \bigg)\\ & - c_i (x^{(1)}_i)^2,
\end{align*}
\begin{equation*}
    u^{(2)}_k (x^{(2)}_{k_{i2}}, \pmb{x}^{(2)}_{I_{k_{i2}}}) = x^{(2)}_p \bigg( \sum_{p \neq k} V_{kp} x^{(2)}_p \bigg).
\end{equation*}
We choose the cost coefficients $c_i$ to be large enough so that the $\Upsilon(F)$ satisfies the $P_{\Upsilon}$ condition(from Appendix A, strong diagonal dominance implies $P_{\Upsilon}$ condition). In the experiments, the $\rho(\Gamma)$(Se Appendix A for $\Gamma$) has a value between $[0.7,0.8]$.

Then under the flat perspective, a level-1 agent $a^{(1)}_i$ has the following utility function
\begin{align*}
    u^{flat}_i(x^{(1)}_i, x^{(1)}_{-i}) = &~ b_i x^{(1)}_i + x^{(1)}_i \bigg( \sum_{j \neq i} W^{flat}_{ij} x^{(1)}_j \bigg) \\
    & - c_i (x^{(1)}_i)^2 + d_i,
\end{align*}
where
\begin{equation*}
    d_i = \sum_{j \in I_i} x^{(1)}_j \bigg( \sum_{p \notin S^{(2)}_{k_{i2}}} W^{flat}_{jp} x^{(1)}_p \bigg),
\end{equation*}
\begin{equation*}
    W^{flat} = \begin{bmatrix} 
    W_1 & V_{1,2} \cdot \pmb{1} & \cdots & V_{1,100} \cdot \pmb{1}\\
    V_{2,1} \cdot \pmb{1} & W_2 & \cdots & V_{2,100} \cdot \pmb{1}\\
    \vdots & \vdots & \ddots & \vdots\\
    V_{100,1} \cdot \pmb{1} & \cdots & V_{100,2} \cdot \pmb{1} & W_{100}
    \end{bmatrix},
\end{equation*}
here $\pmb{1}$ represents the all 1 matrix of suitable size(100$\times$100). 

\subsection{General Best Response Games} \label{appendix:nonlinear_data}
For games with general(non-linear) best response, we generated data using the graphical game model similarly like the above. However, this time we use a mixed cost term that is a weighted sum of a quadratic component and an exponential component. Therefore, we can no longer represent the best response functions as linear functions and the best response computing now relies on gradient based optimization steps. In the experiments shown in the main article, the adjacency matrix is generated following $P_{exist}=0.1$, which creates a sparse network. We also tried $P_{exist}=1$ and the results on the dense networks are included in this part of the appendix.

We use $W^{(l)}_i$ to denote the adjacency matrix within $S^{(l)}_i$ and $W^{(L+1)}$ to denote the adjacency matrix between highest level agents. For the 2-level games with general best response, the utility components are set as follows
\begin{align*}
    u^{(1)}_i (x^{(1)}_i, \pmb{x}^{(1)}_{I_i}) = &~ b_i x^{(1)}_i + x^{(1)}_i \bigg( \sum_{j \in I_i} (W^{(2)}_{k_{i2}})_{r_i r_j} x^{(1)}_j \bigg)\\ & - c_i (x^{(1)}_i)^2 - e^{ 0.1 x^{(1)}_i},
\end{align*}
\begin{align*}
    u^{(2)}_i (x^{(2)}_i, \pmb{x}^{(2)}_{I_i}) = &~ x^{(2)}_i \bigg( \sum_{j \neq i} (W^{(3)}_{k_{i3}})_{ij} x^{(2)}_j \bigg) \\
    & - |S^{(2)}_i| \cdot e^{ 0.1 x^{(2)}_i/|S^{(2)}_i|}.
\end{align*}

For 3-level games with general best response, the components in level-1 and 2 remain the same, and the level-3 components are 
\begin{align*}
    u^{(2)}_i (x^{(3)}_i, \pmb{x}^{(3)}_{I_i}) = &~ x^{(3)}_i \bigg( \sum_{j \neq i} W^{(4)}_{ij} x^{(3)}_j \bigg) \\
    & - |S^{(1,3)}_i| \cdot e^{ 0.1 x^{(3)}_i/|S^{(1,3)}_i|}.
\end{align*}

For the 2-level games with linear/nonlinear best response, the utility components are set as follows
\begin{align*}
    u^{(1)}_i (x^{(1)}_i, \pmb{x}^{(1)}_{I_i}) = &~ b_i x^{(1)}_i + x^{(1)}_i \bigg( \sum_{j \in I_i} (W^{(2)}_{k_{i2}})_{r_i r_j} x^{(1)}_j \bigg)\\ & - c_i (x^{(1)}_i)^2,
\end{align*}
\begin{align*}
    u^{(2)}_i (x^{(2)}_i, \pmb{x}^{(2)}_{I_i}) = &~ x^{(2)}_i \bigg( \sum_{j \neq i} (W^{(3)}_{k_{i3}})_{ij} x^{(2)}_j \bigg) \\
    & - |S^{(2)}_i| \cdot e^{ 0.1 x^{(2)}_i/|S^{(2)}_i|}.
\end{align*}

Again, the adjacency matrix and the cost terms will be scaled to ensure that Assumption \ref{assumption:unique} holds, and in the experiments, the $\rho(\Gamma)$(Se Appendix A for $\Gamma$) has a value between $[0.7,0.8]$.

Hyperparameter settings: besides the parameters in the graphical games, the parameter $h_i^{(l)}$ in the loss function updates in Eqn (\ref{eqn:lambda_update}) is chosen arbitrarily. These parameters can also be referred to as ``penalty parameters''. In our experiments, the performance over these parameters are rather smooth under assumption \ref{assumption:unique}. The hyperparameters $h_i^{(l)}$ are set to the same value on each level $l$. In the 2-level case, we perform a binary search on these hyperparameter, where each value is tested for 5 runs to see the average performance. For the 3-level case, we need to determine 2 hyperparameter values, and this is done by a fixed step size search performed iteratively on the two values. We tune the first one, each value is tested for 5 runs like the above, while fixing the second value, after that, we switch to the tuning of the second value and this process keeps iteratively. The parameters we used in the numerical experiments are
\begin{itemize}
    \item 2-Level game: $h_i^{(2)} = 0.2, 0.1, 0.06, 0.03, 0.01$; for network sizes $30^2, 50^2, 100^2, 200^2, 500^2$ respectively. With tuning range $[0, 0.5]$.
    \item 3-Level game: 
    
    For SH-BRD: $(h_i^{(2)}, h_j^{(3)}) = (0.65, 0.1)$, $(0.32, 0.03)$, $(0.2, 0.01)$, $(0.12, 0.006)$, $(0.04, 0.003)$; for network sizes $10^3, 20^3, 30^3, 50^3, 100^3$ respectively. With tuning range $[0, 0.5]^2$ and tuning step $0.002$.
    
    For HH-BRD: $h_i^{(sl_1)} = 0.7, 0.3, 0.21, 0.125, 0.063$ for network sizes $10^3, 20^3, 30^3, 50^3, 100^3$ respectively. With tuning range $[0, 0.5]$
\end{itemize}

Under the current parameter settings, we still haven't bring out the best performances of SH-BRD, and HH-BRD. In act, the performance gap between the current setting and the optimal setting won't be too large since the best response steps are well-posed. And even with their sub-optimal performances, we have seen their advantages over other algorithms.

In \citep{He2000ADMAdap}, the authors mentioned an adaptive method to generate the penalty parameter matrix $H$ which is generally not diagonal, that can speed up the problem solving steps. This will be an interesting direction to generalize our current algorithm when the best response functions become more ill-posed in the future. 

\subsection{CPU Specs:}
\begin{itemize}
    \item CPU: 6 cores, 12 threads, 2.60/4.50 GHz, 12MB Cache
    \item OS: Windows 10
    \item Software: Python 3.7
    \item RAM: 16 GB
\end{itemize}

\subsection{Results on Dense Networks}
\begin{table}[htbp]
	\begin{center}
		\begin{tabular}{p{0.02\textwidth} p{0.125\textwidth} p{0.125\textwidth} p{0.125\textwidth}}
			\toprule
			\textbf{Size} & \textbf{BRD} & \textbf{MS-BRD} & \textbf{SH-BRD}\\
			\midrule
			$30^2$ & (2.97$\pm$0.24)$\times 10^7$ & (9.91$\pm$0.81)$\times 10^5$
			& \textbf{(8.31$\pm$0.66)$\times 10^5$}
			\\
			
			\midrule
			$50^2$ & (2.41$\pm$0.22)$\times 10^8$ & (4.83$\pm$0.45)$\times 10^6$
			& \textbf{(3.27$\pm$0.30)$\times 10^6$}
			\\
			
			\midrule
			$100^2$ & (4.07$\pm$0.34)$\times 10^9$ & (4.07$\pm$0.34)$\times 10^7$
			& \textbf{(3.04$\pm$0.22)$\times 10^7$}
			\\
			
			\midrule
			$200^2$ & (6.66$\pm$0.62)$\times 10^{10}$ & (3.33$\pm$0.31)$\times 10^8$
			& \textbf{(2.44$\pm$0.17)$\times 10^8$}
			\\
			
			\midrule
			$500^2$ & (2.72$\pm$0.29)$\times 10^{12}$ & (5.53$\pm$0.49)$\times 10^9$
			& \textbf{(3.26$\pm$0.26)$\times 10^9$}
			\\
			
			\bottomrule
		\end{tabular}
	\end{center}
	
	\caption{Convergence and complexity (flops) comparison  with linear best response under multiple initialization, dense network.}\label{table:num_flops}
\end{table}

\begin{table}[htbp]
	\begin{center}
		\begin{tabular}{p{0.05\textwidth} p{0.11\textwidth} p{0.11\textwidth} p{0.11\textwidth}}
			
			\toprule
			\textbf{Size} & \textbf{BRD} & \textbf{MS-BRD} & \textbf{SH-BRD}\\
			\midrule
			$30^2$ & 0.99$\pm$0.03 & 0.49$\pm$0.02 & \textbf{0.24$\pm$0.01}
			\\
			
			\midrule
			$50^2$ & 22.80$\pm$0.05 & 1.83$\pm$0.06
			& \textbf{0.69$\pm$0.01}
			\\
			
			\midrule
			$100^2$ & 1351$\pm$7 & 13.28$\pm$0.26
			& \textbf{4.70$\pm$0.06}
			\\
			
			\midrule
			$200^2$ & $>18000$ & 159.9$\pm$0.8 & \textbf{58.07$\pm$0.42}
			\\
			
			\midrule
			$500^2$ & nan & 3505$\pm$54 & \textbf{1286$\pm$20}
			\\
			
			\bottomrule
		\end{tabular}
	\end{center}
	
	\caption{CPU times on a single machine on 2-Level games with general best response functions, dense network; All times are in seconds.}\label{table:L2_CPU}
\end{table}

\begin{table}[htbp]
	\begin{center}
		\begin{tabular}{p{0.03\textwidth} p{0.11\textwidth} p{0.11\textwidth} p{0.12\textwidth}}
			
			\toprule
			\textbf{Size} & \textbf{BRD} & \textbf{MS-BRD} & \textbf{SH-BRD}\\
			\midrule
			$30^2$ & 1.63$\pm$0.12 & 0.57$\pm$ 0.02 & \textbf{0.028$\pm$0.002}
			\\
			
			\midrule
			$50^2$ & 30.65$\pm$0.35 & 1.94$\pm$0.03
			& \textbf{0.051$\pm$0.003}
			\\
			
			\midrule
			$100^2$ & 1660$\pm$3 & 13.93$\pm$0.25
			& \textbf{0.33$\pm$0.02}
			\\
			
			\midrule
			$200^2$ & $>18000$ & 163.1$\pm$1.4 & \textbf{1.32$\pm$0.04}
			\\
			
			\midrule
			$500^2$ & nan & 3416$\pm$52 & \textbf{29.37$\pm$0.91}
			\\
			
			\bottomrule
		\end{tabular}
	\end{center}
	
	\caption{CPU times on a single machine for 2-Level, linear/nonlinear best-response games, dense network;  All times are in seconds.}\label{table:Lin+Nonlin_CPU}
\end{table}

\begin{table}[htbp]
	\begin{center}
		\begin{tabular}{p{0.02\textwidth} p{0.085\textwidth} p{0.085\textwidth} p{0.085\textwidth} p{0.085\textwidth}}
			
			\toprule
			\textbf{Size} 
			& \textbf{BRD} & \textbf{MS-BRD} & \textbf{SH-BRD} & \textbf{HH-BRD}\\
			\midrule
			$10^3$ & 1.25$\pm$0.02 & 0.39$\pm$0.01 & 0.57$\pm$0.02 & \textbf{0.34$\pm$0.01}
			\\
			
			\midrule
			$20^3$ & 617.3$\pm$4.7 & 2.85$\pm$0.07 & 4.50$\pm$0.06 & \textbf{2.56$\pm$0.06}
			\\
			
			\midrule
			$30^3$ & $>$ 18000 & 10.25$\pm$0.25 & 17.87$\pm$0.14
			& \textbf{9.53$\pm$0.09}
			\\
			
			\midrule
			$50^3$ & nan & 58.04$\pm$0.32 & 100.8$\pm$0.41
			& \textbf{51.86$\pm$0.24}
			\\
			
			\midrule
			$100^3$ & nan & 926.8$\pm$6.4 & 2131$\pm$11
			& \textbf{780.9$\pm$3.0}
			\\
			
			\bottomrule
		\end{tabular}
	\end{center}
	
	\caption{CPU times in seconds on a single machine on 3-Level, general best response games, dense network; All times are in seconds.}	\label{table:L3_CPU}
\end{table}

We can see that though the results in linear best response games are very different in sparse and dense networks, the results in games with non-linear best responses are quite similar in both types of networks. In games with linear best responses, the standard deviation results from different initialization. For the same game, one initial action profile's distance(measured in Euclidean norm) to the equilibrium point can be 20 times to the distance of another initial action profile. This results in different number of iterations of the algorithm before convergence. However, it only takes about 20\% more iterations for a ``distant'' initial action profile to reach convergence, which shows that these algorithms have good convergence property under Assumptions 1-3. In games with non-linear best responses, the standard deviations of CPU times are relatively small(around 1\%) compared to the mean values, and it shows that the performance of all algorithms are stable with a fixed initial action profile.

\section{Algorithm Performances and Network Sizes} \label{sec:net_size_perf}
In this part, we present some results that show the algorithms' performances with different network sizes in 2-level games. 

Figure \ref{fig:num_flops_per_iter} shows the number of flops per iteration for the three algorithms in $I \times M$ games where $I$ is the number of agents in each group and $M$  the number of groups in the network. Both Algorithms \ref{alg:hierarchical} and \ref{alg:hybrid} outperform Algorithm \ref{alg:flat_BRD}. Algorithm \ref{alg:hybrid} generally has lower complexity per iteration compared to Algorithm \ref{alg:hierarchical} since it has less input in every sub-problem and the number of sub-problems are similar in Algorithm \ref{alg:hierarchical} and \ref{alg:hybrid} when the group sizes are large. However, when  group sizes are small compared to the number of groups, Algorithm \ref{alg:hierarchical} and \ref{alg:hybrid} are similar per iteration.

\begin{figure}[htbp] 
	\centering
    \includegraphics[width=8cm]{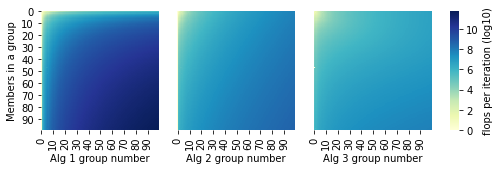} 
	\caption{Complexity per iteration for linear best response.}	\label{fig:num_flops_per_iter}
\end{figure}



%

\section{Reverse Engineer Multi-scale Structure} \label{appendix:reverse_engineer}

A question that naturally arises is whether sparsity in the network can be exploited when the multi-scale structure is not readily available. 
The utility function in Eqn (\ref{eqn:utility}) suggests that such reverse engineering is possible if the game satisfies:
\begin{enumerate}
    \item An agent is either connected to all agents in another group or not connected to any agent in that group; If so, we can create a set of possible group partitions.
    \item Based on the partition in the previous step, agents in one group have the same dependency on an agent in another group.
    \item Based on the partition, we can represent the groups' aggregate actions from their members' actions using some aggregate functions.
    \item Based on the partition, the original utility function of each agent can be separated to components on different levels, each component only based on the actions and dependencies on the corresponding level.
\end{enumerate}

An example of the first condition is shown in Figs. \ref{fig:ungrouped} and \ref{fig:grouped}. For the other conditions, the ``flattened'' utility functions used in Appendix \ref{appendix:numerical_data} are good examples.


\begin{figure}[htbp]
	\centering
	\begin{minipage}{0.2\textwidth}
		\centering
		\includegraphics[width=3.5cm]{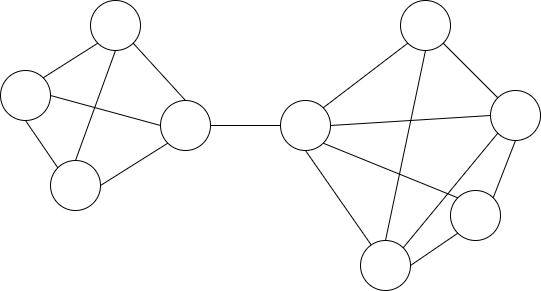} 
		\caption{Ungrouped.}\label{fig:ungrouped}
	\end{minipage}\hfill
	\begin{minipage}{0.2\textwidth}
		\centering
		\includegraphics[width=3.5cm]{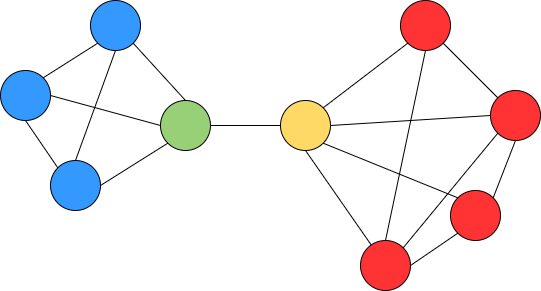} 
		\caption{Grouped.}\label{fig:grouped}
	\end{minipage}
\end{figure}

\section{Flow Charts of the Algorithms}

\begin{figure}[htbp]
	\centering
	\begin{minipage}{0.15\textwidth}
		\centering
		\includegraphics[width=2.8cm]{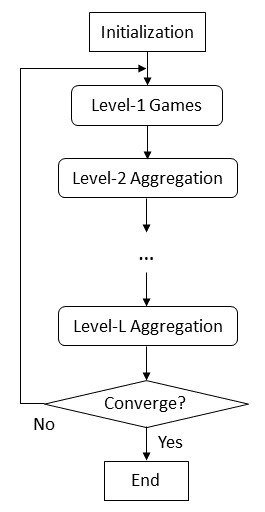} 
		\caption{MS-BRD}\label{fig:alg2_flow}
	\end{minipage}\hfill
	\begin{minipage}{0.15\textwidth}
		\centering
		\includegraphics[width=2.8cm]{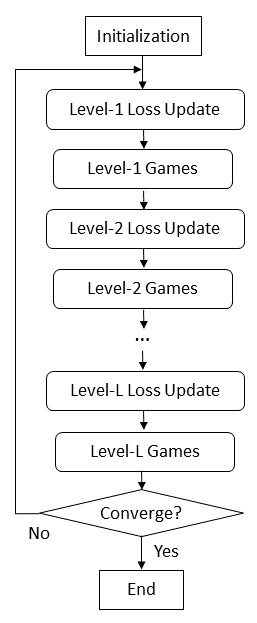} 
		\caption{SH-BRD}\label{fig:alg3_flow}
	\end{minipage}
	\begin{minipage}{0.15\textwidth}
		\centering
		\includegraphics[width=2.8cm]{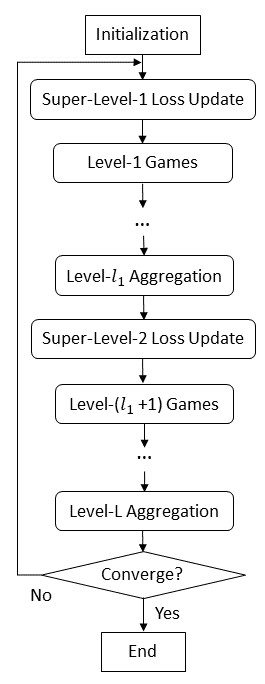} 
		\caption{HH-BRD}\label{fig:alg4_flow}
	\end{minipage}
\end{figure}

\end{document}